%% file: ms.tex
\documentclass{article}
\pdfoutput=1
\input{alex_wang_macros}

\usepackage{biblatex}
\bibliography{hardy}
\newcommand{\psiquantity}{Dirichlet content}
\newcommand{\psitwoquantity}{Neumann content}

\title{Hardy-Muckenhoupt Bounds for Laplacian Eigenvalues}
\author{
Gary L.\ Miller\thanks{
	Work partially supported by NSF CCF-1637523 AitF.}\\
	CMU\\
	\texttt{glmiller@cs.cmu.edu}\\
\and
Noel J.\ Walkington\\
	CMU\\
	\texttt{noelw@cmu.edu}\\
\and
Alex L.\ Wang\thanks{
	Work partially supported by NSF CCF-1637523 AitF.}\\
	CMU\\
	\texttt{alw1@andrew.cmu.edu}
}

\begin{document}
\maketitle
\input{abstract}
\input{introduction}

\input{preliminaries}
\input{weighted_hardy}
\input{dirichlet_on_path}
\input{dirichlet_on_graph}
\input{neumann_on_graph}

\section*{Acknowledgements}
We would like to thank Timothy Chu for many helpful discussions.

\printbibliography
\end{document}

%% file: alex_wang_macros.tex
\usepackage{amssymb, amsmath, amsthm, mathtools,bbold,stmaryrd}
\usepackage[utf8]{inputenc}

    \usepackage{geometry}
      \usepackage{marginnote}
    \usepackage{parskip}

    \usepackage{enumitem}
      \setlist[enumerate]{label=(\roman*),font=\upshape}

    \usepackage{hyperref,url}

    \newtheoremstyle{spacedplain}
      {\parskip}{0pt}{\itshape}{}{\bfseries}{.}{5pt plus 1pt minus 1pt}{}
    \newtheoremstyle{spaceddefinition}
      {\parskip}{0pt}{\normalfont}{}{\bfseries}{.}{5pt plus 1pt minus 1pt}{}
    \newtheoremstyle{spacedremark}
      {\parskip}{0pt}{\normalfont}{}{\itshape}{.}{5pt plus 1pt minus 1pt}{}

    \theoremstyle{spacedplain}
    \newtheorem{theorem}{Theorem}[section]
    
    \newtheorem{lemma}[theorem]{Lemma}
    
    \newtheorem*{theorem*}{Theorem}
    \newtheorem*{corollary*}{Corollary}
    \newtheorem*{lemma*}{Lemma}
    \newtheorem*{claim*}{Claim}

    \theoremstyle{spaceddefinition}
    \newtheorem{definition}[theorem]{Definition}

    \newtheorem*{definition*}{Definition}
    \newtheorem*{example*}{Example}
    \newtheorem*{algorithm*}{Algorithm}

    \theoremstyle{spacedremark}
    \newtheorem{remark}[theorem]{Remark}
    
    \newtheorem{fact}[theorem]{Fact}
    \newtheorem*{remark*}{Remark}
    \newtheorem*{question*}{Question}
    \newtheorem*{fact*}{Fact}

    \newcommand{\framedheader}[3]{
      \framebox[\textwidth]{
        \vbox{
          \vspace{2mm}
          \hbox to \textwidth {\hspace{1em}\today, #1 \hfill #2\hspace{1em}}
          \vspace{4mm}
          \hbox to \textwidth {\hfill \Large{#3} \hfill}
          \vspace{2mm}
        }
      }
      \vspace*{4mm}
    }

    \let\emptyset\varnothing
    \newcommand{\set}[1]{\left\{#1\right\}}
    \newcommand{\smid}{\,\middle|\,}

    \newcommand{\NN}{\mathbb{N}}

    \newcommand{\RR}{\mathbb{R}}

    \newcommand{\by}{\times}
    \renewcommand{\restriction}{\mathord{\upharpoonright}}

    \DeclareMathOperator{\spann}{span}

    \DeclareMathOperator{\diag}{diag}

    \newcommand{\mc}{\mathcal}

    \DeclareMathOperator{\supp}{supp}
    
    \def\XXint#1#2#3{{\setbox0=\hbaox{$#1{#2#3}{\int}$ }
    \vcenter{\hbox{$#2#3$ }}\kern-.6\wd0}}



%% file: abstract.tex
\begin{abstract}

We present two graph quantities $\Psi(G,S)$ and $\Psi_2(G)$ which give constant factor estimates to the Dirichlet and Neumann eigenvalues, $\lambda(G,S)$ and $\lambda_2(G)$, respectively. Our techniques make use of a discrete Hardy-type inequality.

\end{abstract}

%% file: introduction.tex
\section{Introduction}
Let $G = (V,E,\mu,\kappa)$ be a vertex and edge
weighted undirected connected graph, i.e. $(V,E)$ forms a connected graph and $\kappa$, $\mu$ are positive weight functions on the edges and vertices respectively. We will think of our graphs as spring mass systems where vertex $v$ has mass $\mu(v)$ and edge $e$ has spring constant $\kappa(e)$.
Let $A$ be the weighted adjacency matrix, let $D$ be the weighted degree matrix, and let $L = D-A$ be the Laplacian matrix.
Let $M$ be the diagonal mass matrix.
Then, the generalized eigenvalues of $L$ with respect to $M$ have a nice interpretation. Specifically, solutions of the generalized eigenvalue problem
\begin{align*}
Lx = \lambda Mx
\end{align*}
correspond to modes of vibration of the spring mass system. When the spring mass system is connected, $\lambda_2$ is the fundamental mode of vibration\footnote{
The quantity $\lambda_2$ is referred to in the literature under various names: the algebraic connectivity, the Fiedler value, the fundamental eigenvalue, etc. In this paper we will refer to $\lambda_2$ as the Neumann eigenvalue to emphasize the boundary assumptions and to parallel our development in the Dirichlet case.
}.
For an introduction to spring mass systems and the Laplacian, see chapter 5 of \cite{strang2006linear}.

The following result, known as Cheeger's inequality, can be traced back to \cite{alon1985lambda1,cheeger1969lower,dodziuk1984difference}. Define the isoperimetric quantity\footnote{
The quantity $\Phi$ is often
referred to as the conductance of the graph or the Cheeger constant.  In this paper we will refer to $\Phi$ as the isoperimetric constant and reserve the term conductance for the conductance of an edge.}
of $G$ to be
\begin{align*}
\Phi(G) = \min_A\set{\frac{\sum_{e\in E(A,\bar A)} c_e}{\min(\mu(A),\mu(\bar A))}\smid A,\bar A\neq \emptyset}.
\end{align*}
Then we can bound $\lambda_2$ by
\begin{align*}
\frac{\lambda_2}{2}\leq \Phi \leq \sqrt{2\lambda_2 \max_i \frac{d_i}{\mu_i}}.
\end{align*}

In this paper, we introduce the \psitwoquantity of a graph. Roughly, the \psitwoquantity, $\Psi_2(G)$, is the minimum ratio over subsets $A,B\subseteq V$ of the conductance between $A$ and $B$ and the minimum mass of either set. Thus, noting that $\sum_{e\in E(A,\bar A)} c_e$ is the conductance between $A$ and $\bar A$, the isoperimetric constant is roughly equal to the \psitwoquantity\ where the minimization is restricted to sets $A,\bar A$. We will show how to use $\Psi_2(G)$ to give a constant factor estimate of $\lambda_2$. Along the way we will also define the \psiquantity, $\Psi(G,S)$, which allows us to estimate the Dirichlet eigenvalue.
In particular, we prove the following theorems.

\begin{theorem*}
Let $G$ be a vertex and edge weighted connected graph with boundary set $S$, a proper nonempty subset of $V$.
Let $\lambda(G,S)$ be the Dirichlet eigenvalue and let $\Psi(G,S)$ be the Hardy quantity of $G$.
Then
\begin{align*}
\frac{\Psi}{4}\leq\lambda \leq \Psi.
\end{align*}
\end{theorem*}

\begin{theorem*}
Let $G$ be a vertex and edge weighted connected graph. Let $\lambda_2(G)$ be the Neumann eigenvalue and let $\Psi_2(G)$ be the \psitwoquantity\ of $G$. Then,
\begin{align*}
\frac{\Psi_2}{4}\leq\lambda_2 \leq \Psi_2.
\end{align*}
\end{theorem*}

\subsection{Related work}
A very recent independent paper \cite{schild2018schur} introduced a quantity $\rho(G)$ specifically in the case of the normalized Laplacian, i.e., when $M=D$. In this setting, the \psitwoquantity\ $\Psi_2(G)$ is equivalent to the definition of $\rho(G)$ up to constant factors: $\frac{\Psi_2}{2}\leq \rho\leq \Psi_2$. In \cite{schild2018schur}, it is proved that
\begin{align*}
\frac{\rho}{25600} \leq\lambda_2 \leq 2\rho.
\end{align*}
This parallels our Theorem \ref{thm:neu_graph} in the normalized Laplacian case with different constants.

The application of the Hardy-Muckenhoupt inequality to estimating the Dirichlet eigenvalue was noted in \cite{miclo1999example}. In that paper, the authors showed how to bound the Dirichlet eigenvalue on an infinite path graph by the (infinite path analogue of) $\Psi$. Specifically,
\begin{align*}
\frac{\Psi}{4}\leq\lambda \leq 2\Psi.
\end{align*}
This parallels our Theorem \ref{thm:dir_path} in the case of a vertex and edge weighted path graph with different constants.

Other methods for estimating $\lambda_2$ have been proposed.
A method for lower bounding $\lambda_2$ based on path embeddings is presented in \cite{GuLeMi99,GuatteryMiller2000,kahale1997semidefinite}. In this method, a graph with known eigenstructure is embedded into a host graph. Then the fundamental eigenvalue of the host graph can be estimated in terms of the eigenstructure of the embedded graph and the ``distortion'' of the embedding.
For a review of path embedding methods, see the introduction in \cite{GuatteryMiller2000}.

\subsection{Applications}
The Laplacian matrix, and in particular its eigenstructure, finds many applications in computer science, physics, numerical analysis, and the social sciences.
Computing the Neumann eigenvector of a graph has become a standard routine used in image segmentation \cite{shi2000normalized} and clustering \cite{von2007tutorial,ng2002spectral}.
The eigenstructure of the Laplacian is used to model virus propagation in computer networks \cite{wang2003epidemic} and design search engines \cite{brin1998anatomy}.
In numerical analysis and physics, the Laplacian matrix is used to approximate differential equations such as heat flow and the wave equations on meshes \cite{quarteroni2009numerical}.

\subsection{Roadmap}
In section \ref{sec:preliminaries}, we set notation and discuss
background related to weighted graphs, Laplacians, the eigenvalue
problems, and electrical networks.  In section
\ref{sec:weighted_hardy}, we introduce Muckenhoupt's weighted Hardy
inequality.  In section \ref{sec:dir_path}, we introduce the Hardy
quantity and the \psiquantity\ and show how Muckenhoupt's result can be used to bound the Dirichlet eigenvalue on a path graph.  In section \ref{sec:dir_graph}, we
extend the bounds on the Dirichlet eigenvalue from path graphs to
arbitrary graphs.  Finally in section \ref{sec:neu_graph}, we
introduce the two-sided Hardy quantity and the \psitwoquantity\ and extend the bounds on the Dirichlet eigenvalue on a graph to the Neumann eigenvalue on a graph.

%% file: preliminaries.tex

\section{Preliminaries}
\label{sec:preliminaries}
\subsection{Vertex and edge weighted graphs}
Let $G=(V,E,\mu,\kappa)$ be an undirected connected graph with vertex set $V$ and edge set $E$. The mass of vertex $v$ is $\mu_v$ and the conductance\footnote{
As we are dealing with spring mass systems, perhaps it would be better to refer to these quantities as spring constants and compliances. Nonetheless, we have chosen to refer to these quantities as conductances and resistances as this is the terminology most commonly found in the spectral graph theory literature.}
of edge $e$ is $\kappa_e$.
We will assume our graphs are connected and that all masses and conductances are positive.


\subsection{Laplacians}
Let $d_v = \sum_{(u,v)\in E} \kappa_{(u,v)}$ be the degree of vertex $v$ and let $D = \diag(d_1,\dots,d_n)$ be the degree matrix. Let $A\in\RR^{V\by V}$ be the adjacency matrix of $G$, i.e. $A(u,v) = \kappa_{(u,v)}$ if $(u,v)\in E$ and $0$ otherwise. The Laplacian matrix corresponding to $G$ is $L = D-A$. Note that the quadratic form associated with $L$ is
\begin{align*}
x^\top L x = \sum_{(u,v)\in E} \kappa_{(u,v)} (x_u - x_v)^2.
\end{align*}

\subsection{The generalized Laplacian eigenvalue problem}
Let $M$ be the diagonal matrix of masses.
\begin{definition}
The \textbf{Neumann problem} on $G$ is to find
\begin{align*}
\lambda_2(G) = \min_{x\in\RR^V}\set{\frac{x^\top Lx}{x^\top Mx} \smid x^\top M1 = 0,\, x\neq 0}.
\end{align*}
We will refer to the minimum value as the \textbf{Neumann eigenvalue}.
\end{definition}

By the Courant-Fischer min-max principle, we can rewrite this quantity as
\begin{align*}
\lambda_2 &= \min_{W}\max_{x\in W} \frac{x^\top L x}{x^\top M x}
\end{align*}
where $W$ varies over the two dimensional subspaces of $\RR^V$.

At times we will consider the Laplacian eigenvalue problem with extra boundary conditions. This corresponds to fixing the value of $x$ at a given set $S$ of vertices to zero.

\begin{definition}
Let $G$ be a graph and let $S$ be a proper nonempty subset of $V$. The \textbf{Dirichlet problem} on $G$ with boundary set $S$ is to find
\begin{align*}
\lambda(G,S) &= \min_{x\in\RR^V}\set{\frac{x^\top L x}{x^\top M x}\smid x\restriction_S = 0,\,x\neq 0}.
\end{align*}
We will refer to the minimum value as the \textbf{Dirichlet eigenvalue}.
\end{definition}

\begin{remark}
Letting $L_{\bar S}$ be the principal submatrix of $L$ indexed by vertices in $\bar S$ and letting $x_{\bar S}$ be the restriction of $x$ onto the corresponding coordinates, we have
\begin{align*}
L_{\bar S} x_{\bar S}= \lambda(G,S) x_{\bar S}.
\end{align*}
In other words $x_{\bar S}$ is an eigenvector of $L_{\bar S}$. We caution that $x$ itself is, in general, not an eigenvector of $L$.
\end{remark}

\subsection{Graphs as electrical networks and effective resistance}
Given an edge-weighted graph, we can think of its edges as electrical conductors with conductance $\kappa_e$. Thinking of $x\in\RR^V$ as an assignment of voltages to the vertices of our electrical network, we have that
\begin{align*}
x^\top L x = \sum_{(u,v)\in E} \kappa_{(u,v)}(x_u-x_v)^2
\end{align*}
is the ``power dissipated in our system''. Then drawing inspiration from physics, we define the effective resistance between two sets of vertices in terms of the minimum power required to maintain a unit voltage drop.
\begin{definition}
Given nonempty disjoint sets $A,B\subseteq V$, the \textbf{effective resistance between $A$ and $B$}, denoted $R(A,B)$, is the quantity such that
\begin{align*}
\frac{1}{R(A,B)} &= \min_{x\in\RR^V}\set{x^\top L x \smid x\restriction_A =1,\, x\restriction_B = 0}.
\end{align*}
\end{definition}

If $A=\set{a}$ is a single element, we will opt to write $R(a,B)$ instead of the more cumbersome $R(\set{a},B)$. Similarly we will write $R(A,b)$ or $R(a,b)$ where appropriate.

\begin{remark}
When $A=\set{a}$ and $B=\set{b}$ are singleton sets, this definition agrees with the standard definition $R(a,b) = \chi_{a,b} L^+ \chi_{a,b}$. In general, we can define $R(A,B)$ in a different way. Consider contracting all vertices in $A$ to a single Vertex $v_A$ and all vertices to a single vertex $v_B$. Then $R(A,B)$ is the effective resistance between $v_A$ and $v_B$ in the new graph. This is the definition given in \cite{schild2018schur}.
\end{remark}

\subsection{Miscellaneous notation}
In our paper $\NN = \set{1,2,\dots}$ does not contain $0$.

%% file: weighted_hardy.tex

\section{Weighted Hardy inequalities}
\label{sec:weighted_hardy}
The following theorem, due to Muckenhoupt \cite{muckenhoupt1972hardy}, relates the $L_2$ norm of the ``running integral'' of a function to its $L_2$ norm.\footnote{The original theorem deals more generally with $L_p$ norms and Borel measures --- see \cite{muckenhoupt1972hardy}.} We refer to this inequality as the Muckenhoupt-Hardy inequality.

\begin{theorem}\label{thm:Muckenhoupt}[Muckenhoupt 1972]
Let $\mu$, $\kappa$ be functions from $\RR_{\geq 0}$ to $\RR_{>0}$. Let $C$ be the smallest (possibly infinite) constant such that for all $f\in L^1_{\text{loc}}(\RR_{\geq 0})$,
\begin{align*}
\int_0^\infty \mu(x) \left(\int_0^xf(t)\,dt\right)^2\,dx &\leq C\int_0^\infty \kappa(x)f(x)^2\,dx.
\end{align*}
Let
\begin{align*}
B = \sup_{r>0} \left(\int_r^\infty \mu(x)\,dx \right)\left(\int_0^r \frac{1}{\kappa(x)}\,dx\right).
\end{align*}
Then $B\leq C\leq 4B$. In particular, $C$ is finite if and only if $B$ is finite.
\end{theorem}

Letting $f=\frac{d}{dx} g$ for some function $g$ with $g(0)=0$ and dividing through by the constant $C$ and the term on the left, we can reinterpret the Muckenhoupt-Hardy inequality as a bound on the Dirichlet eigenvalue on the nonnegative line.
In the next section we will make this statement formal and give a proof of the rephrased theorem in the finite, discrete case. Our proof will be stated in the language of graph Laplacians but closely follows the structure of \cite{miclo1999example,muckenhoupt1972hardy} and is only included for completeness.

%% file: dirichlet_on_path.tex
\section{The Dirichlet problem on path graphs}
\label{sec:dir_path}
Throughout this section, let $G=(V,E,\mu,\kappa)$ be a vertex and edge weighted connected path graph. Let the vertices be $v_0,v_1,\dots,v_N$ and let the boundary set be $S = \set{v_0}$. Let $E=\set{(v_i,v_{i-1})\smid i\in[N]}$ and let edge $(v_i,v_{i-1})$ have conductance $\kappa_i$. Let vertex $v_i$ have mass $\mu_i$.

\subsection{The Hardy quantity and the \psiquantity}
For $A\subseteq V$, let $\mu(A) = \sum_{v_i\in A} \mu_i$.

Let $A\subseteq V\setminus S$ be a set of vertices disjoint from the boundary.
Consider the graph consisting of two vertices $v_S,v_A$ and let the boundary set be $\set{v_S}$. Let $v_A$ have mass $\mu(A)$ and let the edge $(v_S,v_A)$ has conductance $R(S,A)^{-1}$. Then the Dirichlet eigenvalue of this two node system is given by $\frac{R(S,A)^{-1}}{\mu(A)}$. We will define the \psiquantity\ $\Psi$ to be the minimum such quantity and, for historical reasons, we will define the Hardy quantity to be $H=\Psi^{-1}$.

\begin{definition}
Define the \textbf{\psiquantity}, $\Psi$, to be
\begin{align*}
\Psi = \min_{A\subseteq V}\set{\frac{R(S,A)^{-1}}{\mu(A)} \smid A\neq \emptyset,\, A\cap S = \emptyset}.
\end{align*}
\end{definition}
\begin{definition}
Define the \textbf{Hardy quantity} to be $H = \Psi^{-1}$, i.e.
\begin{align*}
H = \max_{A\subseteq V}\set{R(S,A)\mu(A) \smid A\neq \emptyset,\, A\cap S = \emptyset}.
\end{align*}
\end{definition}

In a path graph, we may choose to optimize over tail sets. This gives us a second characterization of $H$ (and thus $\Psi$) on path graphs.
\begin{lemma}
Let $A_k = \set{v_i\smid i\geq k}$ be the tail set beginning at $v_k$. Then
\begin{align*}
H &= \max_{1\leq k\leq N} \left(\sum_{i=1}^k \frac{1}{\kappa_i}\right)\mu(A_k).
\end{align*}
\end{lemma}
\begin{proof}
Let $A\subseteq V\setminus S$. Let $k = \min\set{i\smid v_i\in A}$ be the minimum element in $A$. Then $R(S,A) = R(S,A_k)$ and $\mu(A_k)\geq \mu(A)$. Note also that on a path graph $R(S,A_k) = \sum_{i=1}^k \kappa_i^{-1}$.
\end{proof}

\subsection{Bounding the Dirichlet eigenvalue}

\begin{theorem}\label{thm:dir_path}
Let $G$ be a vertex and edge weighted connected path graph. Let $\lambda(G,v_0)$ be the Dirichlet eigenvalue and let $H(G,v_0)$ be the Hardy quantity of $G$. Then,
\begin{align*}
\frac{1}{4H}\leq\lambda \leq \frac{1}{H}.
\end{align*}
\end{theorem}

We reiterate that the below proof has been known since \cite{muckenhoupt1972hardy} and is included only for completeness.
\begin{proof}
We begin by proving the upper bound. Note that if $x\restriction_A = 1$, then $x^\top Mx \geq \mu(A)$. Applying this bound to $\lambda$, we note that the numerator of the Rayleigh quotient becomes an effective resistance term.
\begin{align*}
\lambda &= \min_x\set{ \frac{x^\top Lx}{x^\top Mx}\smid x_0 = 0,\, x\neq 0}\\
&\leq \min_{1\leq k\leq N} \min_x\set{ \frac{x^\top Lx}{x^\top Mx}\smid x_0 = 0,\,x\restriction_{A_k}=1}\\
&\leq \min_{1\leq k\leq N} \frac{1}{\mu(A_k)} \min_x\set{x^\top Lx\smid x_0=0,\,x\restriction_{A_k} = 1}\\
&= \min_{1\leq k\leq N} \frac{R(S,A_k)^{-1}}{\mu(A_k)}\\
&= H^{-1}.
\end{align*}

On the other hand, let $x$ be an arbitrary nonzero vector with $x_0 = 0$. Applying Cauchy-Schwarz to the voltage drops,
\begin{align*}
\sum_{i=1}^N \mu_i x_i^2 &= \sum_{i=1}^N \mu_i \left(\sum_{j=1}^i (x_j -x_{j-1})\right)^2\\
&= \sum_{i=1}^N \mu_i \left(\sum_{j=1}^i (x_j -x_{j-1})\kappa^{1/2}_jR(x_0,x_j)^{1/4}\frac{1}{\kappa^{1/2}_jR(x_0,x_j)^{1/4}}\right)^2\\
&\leq \sum_{i=1}^N \mu_i \sum_{j=1}^i (x_j -x_{j-1})^2\kappa_jR(x_0,x_j)^{1/2}\sum_{j=1}^i\frac{1}{\kappa_jR(x_0,x_j)^{1/2}}.
\end{align*}

We use the following inequality: for $0<\alpha\leq \beta$, $\frac{1}{\beta}(\beta^2-\alpha^2) = \frac{\beta+\alpha}{\beta}(\beta-\alpha) \leq 2(\beta-\alpha)$. Note that $\frac{1}{\kappa_j} = R(x_0,x_j) - R(x_0,x_{j-1})$, thus the above inequality allows us to bound the second summation as a telescoping series.
\begin{align*}
\sum_{j=1}^i\frac{1}{\kappa_jR(x_0,x_j)^{1/2}} &\leq 2\sum_{j=1}^i R(x_0,x_j)^{1/2} - R(x_0,x_{j-1})^{1/2}\\
&= 2R(x_0,x_i)^{1/2}.
\end{align*}
Comparing this to the Hardy quantity, we have that $R(x_0,x_i)\leq \frac{H}{\mu(A_i)}$. We complete the bound of the original expression by substituting in our estimate of the second summation, switching the order of summation, then applying $\frac{1}{\beta}(\beta^2 - \alpha^2)\leq 2(\beta-\alpha)$ a second time.
\begin{align*}
\sum_{i=1}^N \mu_ix_i^2 &\leq 2H^{1/2} \sum_{i=1}^N \frac{\mu_i}{\mu^{1/2}(A_i)} \sum_{j=1}^i (x_j -x_{j-1})^2\kappa_jR(x_0,x_j)^{1/2}\\
&= 2H^{1/2} \sum_{j=1}^N (x_j -x_{j-1})^2\kappa_jR(x_0,x_j)^{1/2} \sum_{i=j}^N \frac{\mu_i}{\mu^{1/2}(A_i)}\\
&\leq 4H^{1/2}\sum_{j=1}^N (x_j -x_{j-1})^2\kappa_jR(x_0,x_j)^{1/2} \left(\sum_{i=j}^{N-1} \left(\mu^{1/2}(A_i)-\mu^{1/2}(A_{i+1})\right) + \mu^{1/2}_N\right)\\
&\leq 4H^{1/2} \sum_{j=1}^N (x_j -x_{j-1})^2\kappa_jR(x_0,x_j)^{1/2}\mu^{1/2}(A_j)\\
&\leq 4H\sum_{j=1}^N \kappa_j (x_j -x_{j-1})^2.
\end{align*}
Rearranging, we have that for all $x\in\RR^V$ with $x_0 = 0$,
\begin{align*}
\frac{1}{4H}\leq \frac{\sum_{i=1}^N \kappa_i (x_i -x_{i-1})^2}{\sum_{i=1}^N \mu_i x_i^2}.
\end{align*}
Then minimizing over such $x$ concludes the proof.
\end{proof}

The following theorem follows as a corollary.
\begin{theorem}\label{thm:dir_path}
Let $G$ be a vertex and edge weighted connected path graph. Let $\lambda(G,v_0)$ be the Dirichlet eigenvalue and let $\Psi(G,v_0)$ be the \psiquantity\ of $G$. Then,
\begin{align*}
\frac{\Psi}{4}\leq\lambda \leq \Psi.
\end{align*}
\end{theorem}

%% file: dirichlet_on_graph.tex
\section{The Dirichlet problem on general graphs}
\label{sec:dir_graph}
Throughout this section, let $G=(V,E,\mu,\kappa)$ be a vertex and edge weighted connected graph. Let the boundary set, $S$, be a proper nonempty subset of $V$.

\subsection{Bounding the Dirichlet eigenvalue}

\begin{theorem}\label{thm:dir_graph}
Let $G$ be a vertex and edge weighted connected graph with boundary set $S$, a proper nonempty subset of $V$.
Let $\lambda(G,S)$ be the Dirichlet eigenvalue and let $H(G,S)$ be the Hardy quantity of $G$.
Then
\begin{align*}
\frac{1}{4H}&\leq \lambda \leq \frac{1}{H}.
\end{align*}
\end{theorem}
The proof of the upper bound in the graph case is the same as the proof of the upper bound in the path case.
\begin{proof}[Proof of upper bound.]
Note,
\begin{align*}
\lambda &= \min_x\set{ \frac{x^\top Lx}{x^\top Mx}\smid x\restriction_S = 0,\, x\neq 0}\\
&\leq \min_{A\subseteq V,\,x}\set{\frac{x^\top Lx}{x^\top Mx}\smid A\neq \emptyset,\,A\cap S = \emptyset,\,x_0 = 0,\,x\restriction_{A_k}=1}\\
&\leq \min_{A\subseteq V}\set{\frac{R(S,A_k)^{-1}}{\mu(A_k)}\smid A\neq \emptyset,\,A\cap S = \emptyset}\\
&= H^{-1}.
\end{align*}
\end{proof}

Before proving the lower bound, we state a useful fact.
\begin{fact}\label{fact:cut_edges}
Let $e=(a,b)$ be an edge with conductance $\kappa$ and let $\alpha_1,\dots,\alpha_k>0$ such that $\sum_{i=1}^k \alpha_i = 1$. Consider splitting the edge $e$ into $k$ segments, $e_1,\dots,e_k$, with conductance $\kappa(e_i) = \frac{\kappa}{\alpha_i}$ by inserting $k-1$ zero mass vertices.
Let $G$ be the original graph and let $G'$ be the new graph.
Then $\lambda(G,S) = \lambda(G',S)$. In particular, given $x\in\RR^V$ let $y\in \RR^{V'}$ be the linear extension of $x$, then $x^\top L x = y^\top L' y$.
\end{fact}

To prove the lower bound, we use the above fact to reduce the Dirichlet problem on a graph to the Dirichlet problem on a path.
\begin{proof}[Proof of lower bound.]
We construct a new graph $G'=(V',E',\mu',\kappa')$ from $G$ as follows.
Let $x$ be a solution to the Dirichlet problem corresponding to $\lambda(G,S)$. Let $l_0<\dots<l_N$ be the distinct values of $x$. Without loss of generality, suppose $l_0 = 0$.
For each edge $(a,b)\in E$ such that $x_a=l_i<l_{i+1}<l_j = x_b$, split $e$ into $j-i$ segments such that in the minimum energy extension of $x$, the new vertices on $e$ take on all intermediate values $l_{i+1},\dots,l_{j-1}$ (this is possible by Fact \ref{fact:cut_edges}).
Let $y$ be the minimum energy extension of $x$.

Let $\tilde v_i = \set{v\in V'\smid y_v = l_i}$, let $\tilde A_k=\set{v\in V' \smid y_v \geq l_k}$.
Let $\tilde\kappa_i = \sum_{u\in \tilde v_i,\, v\in \tilde v_{i-1}} \kappa'_{(u,v)}$ be the conductance between $\tilde v_i$ and $\tilde v_{i-1}$. Let $\tilde\mu_i = \mu'(\tilde v_i)$.
Then applying Theorem \ref{thm:dir_path},
\begin{align*}
\lambda(G,S) &= \lambda(G',S)\\
&= \min_{z\in\RR^N}\set{\frac{\sum_{i=1}^N \tilde\kappa_i(z_i-z_{i-1})^2}{\sum_{i=1}^N\tilde\mu_i z_i^2}\smid z_0 = 0,\, z\neq 0}\\
&\geq \frac{1}{4} \min_{1\leq k\leq N} \frac{1}{\mu'\left(\tilde A_k\right) \sum_{i=1}^k \frac{1}{\tilde\kappa_i}}\\
&\geq \frac{1}{4} \min_{1\leq k\leq N} \frac{1}{\mu'\left(\tilde A_k\right) R'\left(S,\tilde A_k\right)}\\
&\geq \frac{1}{4} \min_{A'\subseteq V'}\set{\frac{R'(S,A')^{-1}}{\mu'(A')}\smid A'\neq \emptyset,\, A'\cap S = \emptyset}.
\end{align*}
Finally, let $A = A'\cap V$. Then $\mu(A) = \mu'(A')$ and $R(S,A)\geq R'(S,A')$. Thus,
\begin{align*}
\lambda(G,S)&\geq \frac{1}{4} \min_{A'\subseteq V'}\set{\frac{R'(S,A')^{-1}}{\mu'(A')}\smid A'\neq \emptyset,\, A'\cap S = \emptyset}\\
&\geq \frac{1}{4} \min_{A\subseteq V}\set{\frac{R(S,A)^{-1}}{\mu(A)}\smid A\neq \emptyset,\, A\cap S = \emptyset}\\
& = \frac{1}{4H}.
\end{align*}
\end{proof}

The following theorem follows as a corollary.
\begin{theorem}
Let $G$ be a vertex and edge weighted connected graph with boundary set $S$, a proper nonempty subset of $V$.
Let $\lambda(G,S)$ be the Dirichlet eigenvalue and let $\Psi(G,S)$ be the Hardy quantity of $G$.
Then
\begin{align*}
\frac{\Psi}{4}\leq\lambda \leq \Psi.
\end{align*}
\end{theorem}

%% file: neumann_on_graph.tex

\section{The Neumann problem on general graphs}
\label{sec:neu_graph}
Throughout this section, let $G=(V,E,\mu,\kappa)$ be a vertex and edge weighted connected graph.

\subsection{The two-sided Hardy quantity and the \psitwoquantity}

Let $A,B\subseteq V$ be disjoint nonempty sets.
Consider the graph consisting of two vertices $v_A,v_B$ where vertex $v_A$ has mass $\mu(A)$, vertex $v_B$ has mass $\mu(B)$ and the edge $(v_A,v_B)$ has conductance $R(A,B)^{-1}$. Then the Neumann eigenvalue of this two node system is given by $\frac{\mu(A)^{-1} + \mu(B)^{-1}}{R(A,B)}$. We will define the \psitwoquantity\ $\Psi_2$ to be the minimum such quantity and, for historical reasons, we will define the two-sided Hardy quantity to be $H_2=\Psi_2^{-1}$.

\begin{definition}
Define the \textbf{\psitwoquantity} $\Psi_2$ to be
\begin{align*}
\Psi_2 = \min_{A,B\subseteq V}\set{\frac{\mu(A)^{-1} + \mu(B)^{-1}}{R(A,B)} \smid A,B\neq \emptyset,\, A\cap B = \emptyset}.
\end{align*}
\end{definition}
\begin{definition}
Define the \textbf{two-sided Hardy quantity} to be $H_2 = \Psi_2^{-1}$, i.e.
\begin{align*}
H_2 = \max_{A,B\subseteq V}\set{R(S,A)\left(\mu(A)^{-1} + \mu(B)^{-1}\right)^{-1} \smid A,B\neq \emptyset,\, A\cap B = \emptyset}.
\end{align*}
\end{definition}

We note that the isoperimetric constant $\Phi$ of a weighted graph is closely related to $\Psi_2$. Recall
\begin{align*}
\Phi(G) =\min_{A\subset V} \set{\frac{\sum_{e\in E(A,\bar A)} \kappa_e}{\min(\mu(A),\mu(\bar A))}\smid A,\bar A\neq \emptyset}.
\end{align*}
Noting that $\sum_{e\in E(A,\bar A)} c(e) = R(A,\bar A)^{-1}$ and $(\min(\mu(A),\mu(\bar A))^{-1} = \max (\mu(A)^{-1}, \mu(\bar A)^{-1})$, we can rewrite
\begin{align*}
\Phi(G) =\min_{A\subset V} \set{\frac{\max (\mu(A)^{-1}, \mu(\bar A)^{-1})}{R(A,\bar A)}\smid A,\bar A\neq \emptyset}.
\end{align*}
Thus, up to constant factors, $\Phi$ can be thought of as the \psitwoquantity\ where $A$ and $B$ are required to partition the vertices.

\subsection{Bounding the Neumann eigenvalue}

In this section we show how to extend the bounds on the Dirichlet eigenvalue to the Neumann eigenvalue.

We will bound the Neumann eigenvalue by applying Courant-Fischer to a carefully chosen two-dimensional subspace. In particular, we will split our graph into two parts sharing a common boundary. We will then take our two-dimensional subspace to be the linear span of solutions to the Dirichlet problem on either side of this boundary. 

Let $f\in\RR^V$ such that $f$ takes on both positive and negative values. We will write this concisely as $\pm f\notin \RR^V_{\geq 0}$.
We will ``pinch'' the graph at the zero level set of $f$ to create a new graph $G'=(V',E',\mu',\kappa')$:
for every edge $(u,v)\in E$ such that $f_u < 0 < f_v$, insert a new vertex $s$ such that the minimum energy extension of $f$ assigns $f(s) = 0$. Let $\mu'(s) = 0$.

Abusing notation we will also let $f\in\RR^{V'}$ be the minimum energy extension of $f$ to $V'$.
Let $F_0 = \set{v\in V' \smid f_v = 0}$, let $F_{\geq 0} = \set{v\in V'\smid f_v\geq 0}$ and $F_{\leq 0} = \set{v\in V'\smid f_v\leq 0}$.
Similarly define $F_{>0}, F_{<0}$ and note that $G'$ has no edges between $F_{>0}$ and $F_{<0}$.

We have the following lemma regarding the optimal ``pinch.''
\begin{lemma}
\label{lem:graph_pinching}
\begin{align*}
\lambda_2(G) &= \min_{f}\set{\max\left(\lambda(G',F_{\leq 0}),\lambda(G',F_{\geq 0})\right)\smid\pm f\notin \RR^V_{\geq 0}}.
\end{align*}
\end{lemma}
\begin{proof}
Let $\mc R$ denote the quantity on the right hand side.

We begin by showing that $\lambda_2(G) \leq \mc R$.
Let $f\in\RR^V$ take on both positive and negative values. Note that $\lambda_2(G)=\lambda_2(G')$. Let $y,z\in\RR^{V'}$ be solutions to the two Dirichlet problems with Dirichlet eigenvalues $\lambda(G',F_{\leq 0})$ and $\lambda(G',F_{\geq 0})$ respectively.
Note that $\supp(L'z)\subseteq F_{\leq 0}$ and that $y\restriction_{F_{\leq 0}} = 0$, thus $y^\top L' z = 0$. 
Applying Courant-Fischer to the subspace generated by $y$ and $z$,
\begin{align*}
\lambda_2(G) &= \lambda_2(G')\\
&\leq \max_{x\in \spann(y,z)} \frac{x^\top L'x}{x^\top M'x}\\
&= \max_{(\alpha,\beta)\neq 0} \frac{\alpha^2 y^\top L' y + \beta^2 z^\top L'z}{\alpha^2 y^{\top}M'y +\beta^2 z^{\top}M'z }\\
&= \max\left(\lambda(G',F_{\leq 0}),\lambda(G',F_{\geq 0})\right).
\end{align*}

Next we show that $\mc R\leq \lambda_2(G)$. We will exhibit a choice of $f$ taking on both positive and negative values such that $\lambda(G',F_{\leq 0}),\lambda(G',F_{\geq 0}) \leq \lambda_2(G)$. This will additionally imply that the minimum is achieved.

Let $x$ be a solution to the Neumann problem of $G$. We will pick $f=x$. Abusing notation, also let $x\in\RR^{V'}$ be the minimum energy extension of $x$ to $V'$.
Note that $x\restriction_{F_0} = 0$.
Let $y,z$ be $x$ with the $F_{\leq 0}$ and $F_{\geq 0}$ coordinates zeroed out respectively. Note that $L'y$ agrees with $L'x=\lambda_2(G)M'x$ on the support of $y$ and that $y$ agrees with $x$ on the support of $y$. Thus $y^\top L' y = \lambda_2(G) y^\top M' x = \lambda_2(G) y^\top M'y$. Then,
\begin{align*}
\lambda(G',F_{\leq 0}) &\leq \frac{y^\top L'y}{y^\top M'y}\\
&= \lambda_2(G).
\end{align*}
Similarly, $\lambda(G',F_{\geq 0}) \leq \lambda_2(G)$. 
\end{proof}

\begin{lemma}
\label{lem:res_sum}
Let $A\subseteq F_{<0}$ and $B\subseteq F_{>0}$. Then,
\begin{align*}
R'(A,F_0) + R'(B,F_0) \leq R'(A,B).
\end{align*}
\end{lemma}
\begin{proof}
Let $y\in\RR^{V'}$ be an assignment of voltages such that $y\restriction_A = -1$, $y\restriction_{F_0}=0$ and $y^\top Ly = R'(A,F_0)^{-1}$. Let $Y=y^\top Ly$

Let $z\in\RR^{V'}$ be an assignment of voltages such that $y\restriction_B = 1$, $z\restriction_{F_0}=0$ and $z^\top Lz = R'(B,F_0)^{-1}$. Let $Z=z^\top Lz$.

Note that $y$ is zero on $F_{\geq 0}$ and $\supp(L'z)\subseteq F_{\geq 0}$, thus $y^\top L' z = 0$.

Let $\alpha= \frac{Z}{Y+Z}$. Note that $\alpha y + (1-\alpha) z$ is an assignment of voltages with a voltage drop of $1$ across $A$ and $B$. Thus
\begin{align*}
\frac{1}{R(A,B)} &\leq (\alpha y + (1-\alpha) z)^\top L' (\alpha y + (1-\alpha) z)\\
&= \alpha^2 Y + (1-\alpha)^2 Z\\
&= \frac{Y^2Z+ YZ^2}{(Y+Z)^2}\\
&= \frac{YZ}{Y+Z}\\
&= \frac{1}{R'(A,F_0) + R'(F_0,B)}.
\end{align*}
Rearranging terms completes the proof.
\end{proof}

\begin{theorem}\label{thm:neu_graph}
Let $G$ be a vertex and edge weighted connected graph.
Let $\lambda_2(G)$ be the Neumann eigenvalue and let $H_2(G)$ be the two-sided Hardy quantity\ of $G$.
Then
\begin{align*}
\frac{1}{4H_2}&\leq \lambda_2 \leq \frac{1}{H_2}.
\end{align*}
\end{theorem}
\begin{proof}
For $A,B\subseteq V$, $f\in\RR^V$, let $A<_f B$ if $f_a<f_b$ for all $a\in A,\, b\in B$.

We begin by deriving the upper bound. We express $\lambda_2$ in its ``pinch-point'' characterization (Lemma \ref{lem:graph_pinching}), then apply Theorem \ref{thm:dir_graph} to each Dirichlet problem.
\begin{align*}
\lambda_2(G) &= \min_f\set{\max\left(\lambda(G',F_{\leq 0}),\lambda(G',F_{\geq 0})\right)\smid \pm f\notin \RR^V_{\geq 0}}\\
&\leq \min_{A,B\subseteq V,\,f\in\RR^V} \set{\max\left(\frac{R'(A,F_0)^{-1}}{\mu(A)},\frac{R'(B,F_0)^{-1}}{\mu(B)}\right) \smid \pm f\notin \RR^V_{\geq 0},\, A<_f 0 <_f B}.
\end{align*}
Note that given $A,B\subseteq V$, disjoint and nonempty, we can pick $f$, taking both positive and negative values, such that $R'(A,F_0) = \frac{\mu(B)}{\mu(A)+\mu(B)}R'(A,B)$ and $R'(B,F_0) = \frac{\mu(A)}{\mu(A)+\mu(B)}R'(A,B)$. Picking such an $f$, the two terms in the maximum are equal.
\begin{align*}
\lambda_2(G) &\leq \min_{A,B\subseteq V} \set{\frac{\mu(A)^{-1}+\mu(B)^{-1}}{R'(A,B)} \smid A,B\neq\emptyset,\, A\cap B = \emptyset}\\
&= H_2^{-1}.
\end{align*}

Next we derive the lower bound. Again, we express $\lambda_2$ in its ``pinch-point'' characterization and apply Theorem \ref{thm:dir_graph} to each Dirichlet problem.
\begin{align*}
\lambda_2(G) &= \min_{f}\set{\max\left(\lambda(G',F_{\leq 0}),\lambda(G',F_{\geq 0})\right) \smid \pm f\notin \RR_{\geq 0}^V}\\
&\geq\frac{1}{4}\min_{A,B\subseteq V,\,f\in\RR^V} \set{\max\left(\frac{R'(A,F_0)^{-1}}{\mu(A)},\frac{R'(B,F_0)^{-1}}{\mu(B)}\right) \smid \pm f\notin \RR^V_{\geq 0},\, A<_f 0 <_f B}.
\end{align*}

We pull out $\frac{\mu^{-1}(A) + \mu^{-1}(B)}{R'(A,F_0)+R'(B,F_0)}$ from each  term in the maximum and use the following inequality: if $\alpha,\beta>0$, then $\max\left(\frac{1+\alpha}{1+\beta}, \frac{1+\alpha^{-1}}{1+\beta^{-1}}\right)\geq 1$.
\begin{align*}
\lambda_2(G) &=\frac{1}{4}\min_{A,B\subseteq V,\,f\in\RR^V} \set{\max\left(\frac{R'(A,F_0)^{-1}}{\mu(A)},\frac{R'(B,F_0)^{-1}}{\mu(B)}\right) \smid \pm f\notin \RR^V_{\geq 0},\, A<_f 0 <_f B}\\
&= \frac{1}{4}\min_{A,B\subseteq V,\, f\in\RR^V}\set{\frac{\mu(A)^{-1} + \mu(B)^{-1}}{R'(A,F_0) + R'(B,F_0)}\max\left(\frac{1 + \frac{R'(B,F_0)}{R'(A,F_0)}}{1+ \frac{\mu(A)}{\mu(B)}}, \frac{1 + \frac{R'(A,F_0)}{R'(B,F_0)}}{1+ \frac{\mu(B)}{\mu(A)}}\right)\smid \pm f\notin \RR^V_{\geq 0},\, A<_f 0 <_f B}\\
&\geq \frac{1}{4}\min_{A,B\subseteq V,\, f\in\RR^V}\set{\frac{\mu(A)^{-1} + \mu(B)^{-1}}{R'(A,F_0) + R'(B,F_0)}\smid \pm f\notin \RR^V_{\geq 0},\, A<_f 0 <_f B}
\end{align*}
Finally applying Lemma \ref{lem:res_sum}, we have that $R'(A,F_0)+R'(B,F_0)\leq R'(A,B)=R(A,B)$. Thus,
\begin{align*}
\lambda_2(G) &\geq \frac{1}{4}\min_{A,B\subseteq V}\set{\frac{\mu(A)^{-1} + \mu(B)^{-1}}{R(A,B)}\smid A,B\neq\emptyset,\, A\cap B = \emptyset}\\
&= \frac{1}{4H_2}.
\end{align*}
\end{proof}

The following theorem follows as a corollary.
\begin{theorem}
Let $G$ be a vertex and edge weighted connected graph. Let $\lambda_2(G)$ be the Neumann eigenvalue and let $\Psi_2(G)$ be the \psitwoquantity\ of $G$. Then,
\begin{align*}
\frac{\Psi_2}{4}\leq\lambda_2 \leq \Psi_2.
\end{align*}
\end{theorem}